\documentclass[runningheads]{llncs}
\usepackage[T1]{fontenc}
\usepackage{graphicx}
\usepackage{microtype}
\usepackage{makecell}
\usepackage{amssymb}
\usepackage{bbold}
\usepackage{tikz}
\usepackage{amsmath,amssymb,amsfonts}
\usepackage{enumitem}
\usepackage[none]{hyphenat}
\usepackage{pgfplots}
\pgfplotsset{compat=1.18}

\usepackage{listings}

\lstdefinelanguage{json}{
  basicstyle=\ttfamily\scriptsize,
  showstringspaces=false,
  breaklines=true
}

\usetikzlibrary{automata, positioning, arrows}
\tikzset{
	node distance=3cm,
	initial text=$ $
}

\newcommand{\bb}{\mathbb{b}}
\newcommand{\DD}{\mathbb{D}}
\newcommand{\EE}{\mathbb{E}}
\newcommand{\F}{\mathcal{F}}
\newcommand{\FF}{\mathbb{F}}
\renewcommand{\H}{\mathcal{H}}
\newcommand{\RR}{\mathbb{R}}
\renewcommand{\SS}{\mathbb{S}}
\newcommand{\W}{\mathcal{W}}

\newcommand{\best}{\mathrm{best}}
\newcommand{\dom}{\mathrm{dom}}

\newcommand{\lin}{\mathrm{lin}}
\newcommand{\Reg}{\mathrm{Reg}}
\newcommand{\Run}{\mathrm{Run}}
\newcommand{\worst}{\mathrm{worst}}
\newcommand{\wt}{\mathrm{wt}}

\newcommand{\zero}{\mathbb{0}}
\newcommand{\one}{\mathbb{1}}
\newcommand{\eps}{\varepsilon}

\newcommand{\lsem}{[\![}
\newcommand{\rsem}{]\!]}
\newcommand{\sem}[1]{{\lsem #1 \rsem}}
\newcommand{\llangle}{\langle \! \langle}
\newcommand{\rrangle}{\rangle \! \rangle}
\newcommand{\bind}[1]{{\llangle #1 \rrangle}}

\begin{document}
\title{Scenario Constraints with Memory: A Finite-State Approach to Quantitative Financial Analysis}
\titlerunning{Scenario Constraints with Memory}
%
\author{Vitaly N\"urnberg\orcidID{0009-0009-1785-2975}}
\authorrunning{V. N\"urnberg}
%
\institute{Independent Researcher \\
\email{vitaly.nuernberg@gmail.com}}
\maketitle              
\begin{abstract}
Quantifying worst-case and best-case performance under complex market scenarios is a persistent challenge in financial risk management and the verification of path-dependent financial instruments, such as exotic options and structured products. Simulation-based methods are well suited for probabilistic estimation, but they do not directly provide exhaustive guarantees over all admissible scenarios or explicit witnesses for extremal outcomes.

To address this, we introduce a quantitative automata-based framework 
for the exact extremal analysis of financial systems under declarative scenario constraints. At the core of our approach are event history automata (EHAs), a new formal model
that integrates regular-expression event patterns with admissible numerical
intervals to represent constrained event histories with memory.
Quantitative payoffs are represented by weighted finance finite automata (WFFAs),
which allow transition weights to depend on observed market values. By computing the synchronized product of EHAs and WFFAs, our framework enables the exact calculation of upper and lower payoff bounds. Furthermore, the method automatically extracts interpretable witness event histories that realize these extremal outcomes.

We demonstrate the practical viability of the approach through a case study of an
autocallable structured product with path-dependent mechanisms. The case study
analyzes how different scenario constraints affect coupon accumulation, early
redemption, and protection-loss outcomes. Scalability experiments indicate that the framework's execution remains computationally feasible for practical contract horizons and nontrivial constraint configurations. Overall, this approach provides a mathematically rigorous complement to standard financial simulation methods.

\keywords{Formal Methods \and Quantitative Finance \and Weighted Automata \and
Formal Modelling \and Scenario Analysis \and Stress Testing \and Structured
Products \and Case Study}
\end{abstract}
\section{Introduction}

Path dependence is a central feature of many financial instruments and
risk-management problems. In such settings, value and risk may depend not only on
terminal market states but also on the temporal structure of observed events and
the evolution of market quantities over time. 
Examples include path-dependent derivatives such as barrier and Asian
options~\cite{Hul18}, as well as autocallable structured products with early
redemption features~\cite{Gui15}.

In practice, such systems are often analyzed using simulation-based approaches
such as Monte Carlo methods~\cite{Gla04}. These techniques are effective for estimating
expected values and probabilistic risk measures under stochastic assumptions.
However, they explore sampled subsets of the admissible scenario space and therefore
provide limited guarantees about extremal outcomes. Rare scenarios with significant payoff impact may remain undiscovered by sampling.

At the same time, analysts and risk managers often formulate market assumptions in
qualitative and sequential terms: prolonged stress, delayed recovery, repeated
adverse event patterns such as consecutive downward movements, or restrictions on
certain events after specific observation steps; see, e.g.,~\cite{IAA13}.
Such assumptions are naturally expressed as constraints on event histories, but
they are difficult to combine with quantitative payoff models in a way that
supports exact best-case and worst-case analysis.

Formal methods provide a natural foundation for this perspective. Automata and
regular expressions~\cite{HU79,Kle56} are mathematically precise models for
specifying and analyzing discrete event structures. They also form a standard
basis for verification techniques such as model checking~\cite{BK08}. Weighted automata~\cite{DKV09} extend finite automata with quantitative semantics
over semirings. Weighted finance finite automata (WFFAs), introduced recently in~\cite{DN26},
extend this quantitative automata perspective to financial payoff modeling by
allowing payoff contributions to depend on observed market values.

The main goal of this paper is to complement simulation-based financial analysis
with a finite-state framework for rigorous extremal reasoning. Instead of assigning
probabilities to market paths, we describe which paths are admissible under a given
set of scenario assumptions. This non-probabilistic viewpoint is intended for settings such as stress testing, scenario design, and worst-case analysis, where explicit constraints and extremal outcomes are of primary interest; see, e.g.,~\cite{BCBS18}. It also provides a transparent mechanism for incorporating investor views or
expert judgments by expressing qualitative market beliefs as constraints on
admissible scenarios, in the spirit of view-based approaches such as the
Black--Litterman model~\cite{BL92}.

\paragraph{Contributions.}
The main contributions of this paper are as follows:
\begin{itemize}

\item We introduce pattern-based scenario constraints as a declarative specification
language for admissible financial scenarios. These constraints combine two components: event patterns specified by regular expressions and admissible numerical intervals. They allow qualitative market
beliefs, stress assumptions, and temporal restrictions to be expressed directly at
the level of event histories.

\item We define event history automata (EHAs) as finite-state mechanisms for
compiling and resolving pattern-based scenario constraints. EHAs are closely related to Moore machines~\cite{Moo56}: after reading an event history prefix, their state
summarizes the relevant prefix information, and their output determines the
admissible data interval associated with the current history prefix.

\item We develop synchronized product constructions that combine EHAs with the
WFFA-based payoff-modeling formalism of~\cite{DN26}. The resulting product automaton assigns to each admissible finance word the same payoff as the original WFFA, while words violating the EHA scenario constraints are excluded from the domain of evaluation.

\item We derive effective procedures for exact best-case and worst-case payoff
analysis under scenario constraints, adapting graph-based extremal-value techniques
for weighted automata~\cite{ABK22} to WFFAs with data-dependent transition
expressions. These procedures compute exact extremal payoff bounds and automatically extract witness event histories that realize them.

\item We implement the proposed constructions in a Java research prototype
supporting an end-to-end workflow from declarative scenario specification to
scenario-constrained extremal payoff analysis. The prototype translates input
models into automata, constructs the required synchronized products, and computes
extremal payoff values together with witness histories.

\item We evaluate the framework on a stylized autocallable structured-product case
study with typical payoff mechanisms of such products, including coupon
accumulation, early redemption, and downside protection features; see,
e.g.,~\cite{Gui15}. The experiments demonstrate the expressiveness of our declarative scenario specification, the interpretability of extracted witness histories, and the computational scalability of the framework with respect to contract horizon and the number of scenario constraints.

\end{itemize}

\section{Scenario Constraints and Event History Automata}

\subsection{Finance Words and Scenario Languages}

We focus on discrete financial scenarios, observing a market over a finite number
of time steps. At each step, we simultaneously capture a discrete market event and
its corresponding quantitative value. Events represent qualitative triggers, such
as market shocks or corporate earnings announcements, while the associated
numerical observations record market quantities such as asset prices or index
returns at that moment.

Let $\Sigma$ be a finite alphabet of discrete market events, and let $\DD = \RR_{\ge 0}$ be the quantitative data domain.\footnote{More general choices, such as \(\DD=\RR\) or \(\DD=\RR^n\), could model negative returns or several market quantities. For simplicity and clarity of presentation, we restrict attention to the one-dimensional non-negative setting
\(\DD=\RR_{\ge 0}\).}
A \emph{finance word} over $\Sigma$ and $\DD$ is a finite sequence $w = (a_1, d_1) \dots (a_n, d_n)$, where $n \ge 0$, $a_i \in \Sigma$, and $d_i \in \DD$ for all $1 \le i \le n$. The set of all such sequences is denoted by $(\Sigma \times \DD)^*$. These words serve as formal representations of the discrete financial scenarios introduced above.
For instance, the sequence
\(w=(\texttt{earnings\_announcement},100.0)(\texttt{market\_shock},85.5)\)
represents a scenario in which an earnings announcement is observed at a market
value of \(100.0\), followed by a market shock with the market value reduced to
\(85.5\).

In financial analysis, assumptions about future market behavior usually describe
a set of possible scenarios rather than a single future outcome. We represent such
sets as scenario languages. Formally, a \emph{scenario language} \(\mathcal L\)
over \(\Sigma\) and \(\DD\) is a set of finance words, i.e.,
\(\mathcal L \subseteq (\Sigma \times \DD)^*\).

\begin{example}
\label{EX:history_dependent_constraints}
Consider the share price of a publicly traded company observed over a finite
sequence of market events.
We use the alphabet \(\Sigma\) with event labels \texttt{normal},
\texttt{earnings\_miss}, and \texttt{macro\_shock}, denoting an ordinary trading
day, weaker-than-expected company earnings, and an adverse macroeconomic event,
respectively.
By default, admissible prices are required to lie in the interval \([98,102]\).
This default is overridden by the following history-dependent cases, listed in
order of priority:
\begin{itemize}
\item after an \(\texttt{earnings\_miss}\) followed later by a
\(\texttt{macro\_shock}\), the admissible interval becomes \([40,80]\) for the
remainder of the scenario;
\item after a \(\texttt{macro\_shock}\) followed later by an
\(\texttt{earnings\_miss}\), the admissible interval becomes \([90,104]\).
\end{itemize}
The first case has higher priority. Hence, if both ordered patterns occur, as in
\(
\texttt{earnings\_miss}\;\dots\;\texttt{macro\_shock}\;\dots\;
\texttt{macro\_shock}\;\dots\;\texttt{earnings\_miss},
\)
the first case determines the admissible interval, namely \([40,80]\).
These rules define a scenario language: a finance word is admissible if each
observed share price lies in the interval prescribed by the event history prefix.
For example,
\(
(\texttt{normal},100)
(\texttt{earnings\_miss},99)
(\texttt{macro\_shock},70)
(\texttt{normal},d)
\)
is admissible exactly when \(d \in [40, 80]\).
\end{example}

\subsection{Pattern-Based Scenario Constraints}

Example~\ref{EX:history_dependent_constraints} suggests a convenient way to specify scenario assumptions. Relevant event histories are described by regular expressions. Numerical ranges are then assigned to observations whenever such patterns apply. 

\paragraph{Regular expressions.}

Recall that \(\Sigma\) denotes the finite alphabet of discrete market events.
Let \(\Sigma^*\) denote the set of all words \(w=a_1\dots a_n\), where
\(n\ge 0\) and \(a_i\in\Sigma\) for all \(1\le i\le n\). The unique word of
length zero is denoted by \(\eps\). We identify each symbol \(a\in\Sigma\) with
the corresponding one-letter word in \(\Sigma^*\). Any set
\(L\subseteq\Sigma^*\) is called a language over \(\Sigma\).
For languages \(L_1,L_2\subseteq\Sigma^*\), their concatenation is
\(L_1\cdot L_2=\{uv \mid u\in L_1,\ v\in L_2\}\). The Kleene star of
\(L\subseteq\Sigma^*\) is \(L^*=\bigcup_{k \in \mathbb N}L^k\), where
\(L^0=\{\eps\}\) and \(L^{k+1}=L^k\cdot L\) for all $k \in \mathbb N$.

We model event patterns by regular expressions~\cite{Kle56} over \(\Sigma\). Let
\(\Reg(\Sigma)\) be the set of \emph{regular expressions} over \(\Sigma\) generated by
the grammar
\[
E ::= \eps \mid A \mid E + E \mid E \cdot E \mid E^*,
\]
where \(A\subseteq\Sigma\). In particular, \(A=\emptyset\) denotes the empty
language and \(A=\Sigma\) matches any single event.

The \emph{semantics} of a regular expression \(E\in\Reg(\Sigma)\) is the language \(\sem{E}\subseteq\Sigma^*\) defined inductively by \(\sem{\eps}=\{\eps\}\), \(\sem{A}=A\) for \(A\subseteq\Sigma\),
\(\sem{E_1+E_2}=\sem{E_1}\cup\sem{E_2}\),
\(\sem{E_1\cdot E_2}=\sem{E_1}\cdot\sem{E_2}\), and
\(\sem{E^*}=\sem{E}^*\).

\begin{example}
Let \(\Sigma\) be the alphabet from Example~\ref{EX:history_dependent_constraints}. Consider the regular expression
\(
E = \Sigma^* \cdot \texttt{earnings\_miss}\cdot \Sigma^* \cdot
\texttt{macro\_shock} \cdot \Sigma^* .
\)
Its semantics is
$
\sem{E}
=
\Sigma^* \cdot \{\texttt{earnings\_miss}\}\cdot \Sigma^*
\cdot \{\texttt{macro\_shock}\}\cdot \Sigma^* .
$
Thus, \(\sem{E}\) contains exactly those event histories in which an
\(\texttt{earnings\_miss}\) occurs before a later \(\texttt{macro\_shock}\), not necessarily in consecutive steps.
For example,
$
w_1 = \texttt{normal}\cdot\texttt{earnings\_miss}\cdot
\texttt{normal}\cdot\texttt{macro\_shock} \in \Sigma^*
$
satisfies the pattern, i.e., \(w_1 \in \sem{E}\), whereas
$
w_2 = \texttt{normal}\cdot\texttt{macro\_shock}\cdot
\texttt{earnings\_miss} \in \Sigma^*
$
does not, i.e., \(w_2 \notin \sem{E}\).
\end{example}

\paragraph{Numerical intervals.}
When a pattern applies, it restricts the numerical values that are considered possible under the current scenario assumptions. In this paper, we use \(\DD=\RR_{\ge 0}\) and restrict attention to the following interval types: the empty set \(\emptyset\), the full quantitative domain \(\DD\), bounded intervals \([a,b]\) with rational endpoints \(0\le a\le b\), and unbounded intervals \([a,\infty)\) with rational endpoint \(a\ge 0\).
We write \(\mathcal I(\DD)\) for the set of all such intervals.

The empty interval \(\emptyset\) means that any finance word matching the associated pattern is excluded from the scenario language. The full domain \(\DD\) means that there is no numerical restriction. A bounded interval \([a,b]\) imposes lower and upper bounds, while \([a,\infty)\) imposes only a lower bound.

\paragraph{Pattern-based scenario constraints.}
We now combine event patterns and numerical intervals into a rule-based specification language. The goal is to describe how admissible numerical observations depend on the previously observed event history.

A \emph{pattern-based scenario constraint} over \(\Sigma\) and \(\DD\) is a pair
\(
(E,J),
\)
where \(E \in \Reg(\Sigma)\) is an event pattern and \(J \in \mathcal I(\DD)\) is a numerical interval. Its intended meaning is that, whenever the event history matches \(E\), the current numerical observation must lie in \(J\). We write
$
\mathcal C(\Sigma,\DD)
$
for the set of all pattern-based scenario constraints over \(\Sigma\) and \(\DD\).

Several constraints may match the same event history. For example, one pattern may describe a broad market regime, while another describes a more specific stress case. We therefore need a mechanism that selects the interval to be applied.
For a finite set \(X\), we write \(2^X\) for the set of all subsets of \(X\).

\begin{definition}
A \emph{pattern-based scenario constraint system} over \(\Sigma\) and \(\DD\) is a pair
$
\mathcal S=(\mathbb{C},\operatorname{Res}),
$
where \(\mathbb{C}=(c_1,\dots,c_m)\) is a finite sequence of constraints from \(\mathcal C(\Sigma,\DD)\), and
\(
\operatorname{Res}:2^{\{1,\dots,m\}}\to\mathcal I(\DD)
\)
is an \emph{interval resolver}.
\end{definition}

Given an event history \(u\in\Sigma^*\), the set of \emph{active constraint
indices} is
\(
{\operatorname{Act}_{\mathcal S}(u)
=
\{i\in\{1,\dots,m\} \mid u\in \sem{E_i}\},}
\)
where \(c_i=(E_i,J_i)\). Thus, a constraint is active after \(u\) precisely when
its event pattern matches the history \(u\).

The interval resolver then determines which interval is applied for the current
history. The constraint system induces a mapping
\(
\sem{\mathcal S}: \Sigma^*\to\mathcal I(\DD)
\)
defined by
\(
\sem{\mathcal S}(u)
=
\operatorname{Res}\bigl(\operatorname{Act}_{\mathcal S}(u)\bigr).
\)
Hence, every event history \(u\) is assigned the interval of numerical values
admissible after observing \(u\).

For a finance word
\(
w=(a_1,d_1)\dots(a_n,d_n)
\)
and \(i \in \{1,\dots,n\}\), we call
\(
{u_i=a_1\dots a_i} \in \Sigma^*
\)
the \emph{prefix history} at step \(i\). The scenario language defined by
\(\mathcal S\) consists of all finance words whose observed values satisfy the
intervals assigned to their prefix histories:
\[
\mathcal L(\mathcal S)=
\left\{
(a_1,d_1)\dots(a_n,d_n)\in(\Sigma\times\DD)^*
\;\middle|\;
d_i\in \sem{\mathcal S}(a_1\dots a_i)
\text{ for all } 1 \le i \le n
\right\}.
\]

\begin{example}[Priority Resolver]
Scenario specifications often contain both baseline assumptions and more specific
stress conditions. A priority resolver selects the first active constraint in a
given priority order, so that specific stress conditions can override generic
baseline assumptions.

Assume that constraints are ordered by decreasing priority, so that \(c_1\) has
the highest priority. Let \(J_0\in\mathcal I(\DD)\) be a default interval used
when no constraint is active. The priority resolver
\(
\operatorname{Res}_{\mathrm{prio}}^{J_0}:2^{\{1,\dots,m\}}\to\mathcal I(\DD)
\)
is defined as follows. For every non-empty set \(X\subseteq\{1,\dots,m\}\), let
\(
\operatorname{Res}^{J_0}_{\mathrm{prio}}(X)=J_k,
\)
where \(k=\min X\) and \(c_k=(E_k,J_k)\). For the empty set, let
\(
\operatorname{Res}^{J_0}_{\mathrm{prio}}(\emptyset)=J_0.
\)

Thus, among all active constraints, the highest-priority one determines the
interval to be applied. This is useful when stress constraints are intended to
override normal-market assumptions.
\end{example}

\begin{example}[Intersection Resolver]
A conservative alternative is to require all active constraints to be satisfied
simultaneously. This gives the \emph{intersection resolver}
\(\operatorname{Res}_{\cap}:2^{\{1,\dots,m\}}\to\mathcal I(\DD)\), defined by
\(\operatorname{Res}_{\cap}(X)=\bigcap_{i\in X} J_i\), where
\({c_i=(E_i,J_i)}\). For the empty set of active constraints, we set
\(\operatorname{Res}_{\cap}(\emptyset)=\DD\).

Thus, only numerical values contained in every active interval remain admissible.
If the active constraints are inconsistent, the resulting interval is
\(\emptyset\).
\end{example}

\begin{example}
\label{EX:formalization_history_dependent_component}
We now formalize the scenario language described informally in Example~\ref{EX:history_dependent_constraints}.
Let \(\Sigma\) be the alphabet from that example. Consider the following constraints $c_1, c_2 \in \mathcal C(\Sigma, \DD)$, ordered by decreasing priority:
\[
\begin{aligned}
c_1 &= \left(
\Sigma^*\cdot\texttt{earnings\_miss}\cdot\Sigma^*
\cdot\texttt{macro\_shock}\cdot\Sigma^*,
[40,80]
\right), \\
c_2 &= \left(
\Sigma^*\cdot\texttt{macro\_shock}\cdot\Sigma^*
\cdot\texttt{earnings\_miss}\cdot\Sigma^*,
[90,104]
\right).
\end{aligned}
\]
Let \(J_0=[98,102]\) be the default interval and let \(\operatorname{Res}^{J_0}_{\mathrm{prio}}\) be the priority resolver with default interval \(J_0\). This defines the pattern-based scenario constraint system
\(
\mathcal S=((c_1,c_2),\operatorname{Res}^{J_0}_{\mathrm{prio}}).
\)
Consider the event history
\(
u=
\texttt{macro\_shock}\cdot
\texttt{earnings\_miss}\cdot
\texttt{normal}\cdot
\texttt{earnings\_miss}\cdot
\texttt{macro\_shock}\cdot
\texttt{normal}.
\)
For \(i \in \{1,\dots,6\}\), let $a_i \in \Sigma$ be the $i$-th letter of $u$, and let $u_i \in \Sigma^*$ be the prefix history of \(u\) at step \(i\). The assigned intervals are
\(
\sem{\mathcal S}(u_1)=[98,102]
\),
\(
\sem{\mathcal S}(u_i)=[90,104]
\)
for \(i=2,3,4\),
\(
\sem{\mathcal S}(u_i)=[40,80]
\) for \(i=5,6\).
At \(i=2\), the second constraint becomes active. At \(i=5\), both constraints are active, and the first one is selected because it has higher priority. 
\end{example}

\subsection{Event History Automata (EHA)}

Let \(\Sigma\) be a finite event alphabet and let \(\DD = \RR_{\ge 0}\), the data domain. We now
introduce event history automata, a finite-state model that assigns admissible
data intervals to event histories. EHAs can be specified directly, for instance
to represent sequential interval forecasts or tree-like scenario structures, and
they also serve as the compilation target for pattern-based scenario constraints.

\begin{definition}
An \emph{event history automaton} (EHA) over \(\Sigma\) and \(\DD\) is a tuple
\(
\H = (Q,q_0,\delta,\rho),
\)
where \(Q\) is a finite set of states, \(q_0\in Q\) is the initial state, \(\delta:Q\times\Sigma\to Q\) is a deterministic transition function, and \(\rho:Q\to\mathcal I(\DD)\) assigns an interval to each state.
\end{definition}

An EHA can be viewed as a deterministic Moore machine~\cite{Moo56} whose outputs are intervals over the data domain. Each state summarizes the relevant information contained in the event history processed so far, and the output function specifies the admissible numerical values after that history.

For an event history \(u=a_1\dots a_n\in\Sigma^*\), the unique run of \(\H\) on \(u\) is the sequence
\(
q_0,q_1,\ldots,q_n
\)
such that
\(
q_i=\delta(q_{i-1},a_i)
\)
for all 
\(1\le i\le n\).
We write
\(
\operatorname{last}_{\H}(u)=q_n
\)
for the last state of this run. The \emph{behavior} of \(\H\) is the mapping
\(
\sem{\H}:\Sigma^*\to\mathcal I(\DD)
\)
defined for all $u \in \Sigma^*$ by
\(
\sem{\H}(u)
=
\rho\bigl(\operatorname{last}_{\H}(u)\bigr).
\)
The scenario language induced by \(\H\) is
\[
\mathcal L(\H)=
\left\{
(a_1,d_1)\dots(a_n,d_n)\in(\Sigma\times\DD)^*
\;\middle|\;
d_i\in \sem{\H}(a_1\dots a_i)
\text{ for all }1\le i\le n
\right\}.
\]

\begin{example}
Consider the pattern-based scenario constraint system introduced in
Example~\ref{EX:formalization_history_dependent_component}. Here we show how the
first constraint \(c_1\) alone can be modeled by an event history automaton. We use the priority resolver
\(\operatorname{Res}^{J_0}_{\mathrm{prio}}\) with default interval
\(J_0=[98,102]\).
Figure~\ref{fig:eha_earnings_miss_macro_shock} shows the corresponding EHA. It tracks whether an \(\texttt{earnings\_miss}\) has occurred and whether it has later been followed by a \(\texttt{macro\_shock}\).
\begin{figure}[ht]
\centering
\begin{tikzpicture}[
    >=stealth,
    every node/.style={font=\scriptsize},
    state/.style={circle, draw, minimum size=16pt, inner sep=1pt},
]
\node (q0i) at (0,0) {};
\node[state] (q0) at (1.25,0) {$q_0$};
\node[state] (q1) at (4.25,0) {$q_1$};
\node[state] (q2) at (7.25,0) {$q_2$};

\node (q0f) [below=0.0cm of q0] {$[98,102]$};
\node (q1f) [below=0.0cm of q1] {$[98,102]$};
\node (q2f) [below=0.0cm of q2] {$[40,80]$};

\draw[->]
    (q0i) edge[above]  (q0)
    (q0) edge[loop above] node{\texttt{normal}, \texttt{macro\_shock}} (q0)
    (q0) edge[above] node{\texttt{earnings\_miss}} (q1)
    (q1) edge[loop above] node{\texttt{normal}, \texttt{earnings\_miss}} (q1)
    (q1) edge[above] node{\texttt{macro\_shock}} (q2)
    (q2) edge[loop above] node{$\Sigma$} (q2);
\end{tikzpicture}
\caption{Event history automaton for the first constraint \(c_1\) from
Example~\ref{EX:formalization_history_dependent_component}.}
\label{fig:eha_earnings_miss_macro_shock}
\end{figure}
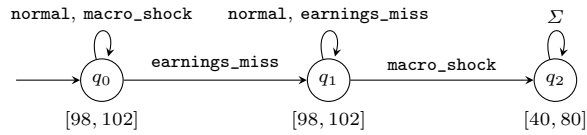
The label \(\Sigma\) on the self-loop of \(q_2\) means that the automaton remains in
\(q_2\) after reading any event from the alphabet. The intervals written below the
states are the values of the output function \(\rho\). Thus, prefixes ending in
\(q_0\) or \(q_1\) are assigned the default interval \([98,102]\), while prefixes
ending in \(q_2\) are assigned the stressed interval \([40,80]\).
\end{example}

The previous example illustrates how a single pattern constraint can be represented by an EHA. The same idea extends to arbitrary finite constraint systems: each pattern is compiled into a deterministic automaton, and the resulting automata are combined by a synchronous product.

\begin{theorem}[Equivalence of pattern constraints and EHAs]
\label{THM:equiv_pattern_constraints_eha}
Let \(\Sigma\) be a finite event alphabet and let \(\DD\) be the quantitative domain.

\begin{enumerate}[label=(\alph*)]
\item For every pattern-based scenario constraint system \(\mathcal S\) over \(\Sigma\) and \(\DD\), there exists an EHA \(\H\) over \(\Sigma\) and \(\DD\) such that
\(
\sem{\H}=\sem{\mathcal S}.
\)

\item For every EHA \(\H\) over \(\Sigma\) and \(\DD\), there exists a pattern-based scenario constraint system \(\mathcal S\) over \(\Sigma\) and \(\DD\) such that
\(
\sem{\mathcal S}=\sem{\H}.
\)
\end{enumerate}
\end{theorem}

For part~(a), each event pattern is translated into a deterministic finite automaton; their
synchronous product tracks the active constraints, and the EHA output is obtained
by applying the resolver. Part~(b) follows by using one pattern for each EHA
state.

\begin{remark}[Complexity]
If the system has \(m\) constraints and each pattern automaton has \(O(N)\)
states, the product construction may have \(O(N^m)\) states in the worst case.
After construction, the interval for a history \(u\in\Sigma^*\) is computed in
time \(O(|u|)\) by following the unique run.
\end{remark}

\section{Quantitative Payoff Evaluation via Weighted Automata}

In this section, we recall the quantitative automata model used for payoff
evaluation. Our presentation is based on weighted finite finance automata
(WFFAs)~\cite{DN26}, which provide a finite-state formalism for modeling
financial payoff structures over event-data sequences.

We only recall the parts of the framework that are needed in the present paper. 
For the general theory, further examples, and the connection with weighted finance regular expressions, we refer the reader to~\cite{DN26}.

\paragraph{Finance Semirings.}

We first recall the algebraic setting used for quantitative payoff evaluation. The role of the algebraic structure is to specify how local payoff contributions are accumulated along a run and how alternative runs are compared.

A \emph{semiring} (cf., e.g.,~\cite{DKV09}) is a structure
\(
\SS=(S,\oplus,\otimes,\zero,\one),
\)
where \((S,\oplus,\zero)\) is a commutative monoid, \((S,\otimes,\one)\) is a monoid, multiplication distributes over addition, and
\(
s\otimes\zero=\zero\otimes s=\zero
\)
for all \(s\in S\).
The main example used in this paper is the \emph{max-plus semiring}
\(
\SS_{\max,+}
=
(\RR\cup\{-\infty\},\max,+,-\infty,0).
\)
Here, \(+\) accumulates values along a path, while \(\max\) selects the best value among alternatives.

To incorporate quantitative market observations, we use an additional data-binding operation. Let \(\DD\) be the data domain. A \emph{data-binding function} over \(\SS\) and \(\DD\) is a mapping
\(
\bb:\DD\times S\to S.
\)
It describes how an externally observed numerical value, such as an asset price, is combined with a semiring value during the evaluation of a transition.

A \emph{finance semiring} is a triple
\(
\FF=(\SS,\DD,\bb).
\)
Thus, \(\SS\) specifies the algebra used for accumulation and optimization,
\(\DD\) specifies the domain of external quantitative observations, and \(\bb\)
specifies the operation by which observations are combined with semiring weights.

In the present paper, we restrict attention to the \emph{linear finance semiring}
\({
\FF_{\lin}
=
(\SS_{\max,+},\RR_{\ge 0},\bb_{\lin}),}
\)
where
\(
\bb_{\lin}(d,s)=d\cdot s
\)
for \(s\in\RR\), and \({\bb_{\lin}(d,-\infty)=-\infty}\) for all \(d \in \DD\). Thus, finite scalar
weights are transformed by multiplication with the observed data value. For instance, if \(d\) is an observed price and \(s\) is the number of units held, then \(\bb_{\lin}(d,s)\) is the market value of the position.

\paragraph{Local Payoff Expressions}

Weighted finance automata evaluate finance words step by step. When the automaton reads a pair \((a,d)\), the event \(a\) determines which transitions are available, while the quantitative value \(d\) is used to evaluate the payoff expression attached to the chosen transition. These expressions are local in the sense that they depend only on the current data value.

\begin{definition}
Let \(\FF=(\SS,\DD,\bb)\) be a finance semiring with
\(\SS=(S,\oplus,\otimes,\zero,\one)\). The set \(\EE_{\FF}\) of
\emph{local payoff expressions} over \(\FF\) is generated by the grammar
\[
e ::= s \mid \bind{s} \mid e\oplus e \mid e\otimes e \mid e=e \mid e\neq e
\qquad (s\in S).
\]
\end{definition}

These expressions correspond to the \(\FF\)-expressions introduced in~\cite{DN26}; we use the term local payoff expressions here to emphasize their role as transition-level payoff specifications.
Let \(e\in\EE_{\FF}\) be a local payoff expression. The \emph{semantics} of \(e\) is the mapping
\(
\sem{e}:\DD\to S
\)
defined inductively as follows. For every \(d\in\DD\),
{\footnotesize
\[
\begin{array}{rcl@{\qquad}rcl}
\sem{s}(d) &=& s
&
\sem{\bind{s}}(d) &=& \bb(d,s)
\\[2mm]
\sem{e_1\oplus e_2}(d) &=& \sem{e_1}(d)\oplus \sem{e_2}(d)
&
\sem{e_1\otimes e_2}(d) &=& \sem{e_1}(d)\otimes \sem{e_2}(d)
\\[2mm]
\sem{e_1=e_2}(d) &=&
\begin{cases}
\one, & \text{if } \sem{e_1}(d)=\sem{e_2}(d),\\
\zero, & \text{otherwise,}
\end{cases}
&
\sem{e_1\neq e_2}(d) &=&
\begin{cases}
\one, & \text{if } \sem{e_1}(d)\neq\sem{e_2}(d),\\
\zero, & \text{otherwise.}
\end{cases}
\end{array}
\]
}
Constants provide fixed semiring values, terms of the form \(\bind{s}\) combine
the current data value with \(s\) through \(\bb\), and comparison expressions act
as guards. For the linear finance semiring \(\FF_{\lin}\), equality and
inequality guards can encode the usual order comparisons \(<,\le,>,\ge\); see
Example~2.9 in~\cite{DN26}.

For example, over \(\FF_{\lin}\), the expression
\(
e=\bigl(\bind{1}\otimes(-100.0)\bigr)
\)
has the semantics
\(
\sem{e}(d)=d-100.0
\)
for all \(d\in\DD\). Hence, the expression
\(
e'=\bigl((\bind{1}\otimes(-100.0))\oplus 0.0 \bigr)
\)
has the semantics
\(
\sem{e'}(d)=\max(d-100.0,0.0).
\)
Here, \(0.0\) is a real-valued constant, not the semiring zero \(\zero=-\infty\).

\paragraph{Weighted Finite Finance Automata (WFFA)}

Let \(\FF=(\SS,\DD,\bb)\) be a finance semiring with \(\SS=(S,\oplus,\otimes,\zero,\one)\).
We now recall weighted finite finance automata from~\cite{DN26}. They extend the weighted automata of Sch\"utzenberger by allowing transition weights to be local payoff expressions from $\EE_{\FF}$ rather than fixed semiring constants from $S$. In this way, transition weights can incorporate the quantitative input value observed at the current step.

\begin{definition}
A \emph{weighted finite finance automaton} (WFFA) over an event alphabet \(\Sigma\) and a finance semiring \(\FF\) is a tuple
\(
\mathcal W=(Q,I,T,F,\wt_I,\wt_T,\wt_F),
\)
where \(Q\) is a finite set of states, \(I,F\subseteq Q\) are the initial and final states,
\(
T\subseteq Q\times\Sigma\times Q
\)
is the transition relation, \(\wt_I:I\to S\) and \(\wt_F:F\to S\) are initial and final weight functions, and
\(
\wt_T:T\to\EE_{\FF}
\)
assigns a local payoff expression to each transition.
\end{definition}
Let
\(
w=(a_1,d_1)\ldots(a_n,d_n)\in(\Sigma\times\DD)^*
\)
be a finance word. A \emph{run} of \(\mathcal W\) on \(w\) is a sequence
\(
\varrho=
(q_0 \xrightarrow{a_1} q_1 \xrightarrow{a_2} \ldots \xrightarrow{a_n} q_n)
\)
such that \(q_0\in I\), \(q_n\in F\), and \(t_i=(q_{i-1},a_i,q_i)\in T\) for all \(1\le i\le n\). The \emph{accumulated payoff} of the run \(\varrho\) on \(w\) is
\(
\wt_{\mathcal W}(\varrho,w)
=
\wt_I(q_0)\otimes
\sem{\wt_T(t_1)}(d_1)\otimes\ldots\otimes
\sem{\wt_T(t_n)}(d_n)\otimes
\wt_F(q_n).
\)
Thus, each transition contributes a local payoff value, and these contributions are accumulated along the run using the semiring multiplication \(\otimes\).

The \emph{behavior} of \(\mathcal W\) is the mapping
\(
\sem{\mathcal W}:(\Sigma\times\DD)^*\to S
\)
defined by
\(
{\sem{\mathcal W}(w)
=
\bigoplus \left(
\wt_{\mathcal W}(\varrho,w) \mid \varrho\in\Run_{\mathcal W}(w) \right),
}
\)
where \(\Run_{\mathcal W}(w)\) is the set of all runs of \(\mathcal W\) on \(w\).
Thus, local payoff expressions are evaluated on the data values of the input word. Their values are accumulated along each run using \(\otimes\), and the accumulated payoffs of all compatible runs are aggregated using \(\oplus\). In the max-plus setting, this means that local payoff values are summed along a run, and the behavior selects the maximal accumulated payoff over all compatible runs.

\paragraph{Scenario-Restricted Payoff Evaluation}
A WFFA assigns a payoff value to every finance word, whereas an EHA specifies which finance words are admissible under the chosen scenario assumptions. We now combine these two components.

A \emph{quantitative finance language}~\cite{DN26} over \(\Sigma\), \(\DD\), and \(\SS\) is a mapping
\(
{\mathcal F:(\Sigma\times\DD)^*\to S.}
\)
It assigns a semiring value to each finance word. In particular, the behavior \(\sem{\mathcal W}\) of a WFFA is a quantitative finance language. Let
\(
\mathcal L \subseteq(\Sigma\times\DD)^*
\)
be a scenario language. The \emph{restriction} of \(\mathcal F\) to \(\mathcal L\) is the quantitative finance language
\(
(\mathcal F\cap \mathcal L):(\Sigma\times\DD)^*\to S
\)
defined by setting \((\mathcal F\cap \mathcal L)(w)=\mathcal F(w)\) for \(w\in \mathcal L\), and
\((\mathcal F\cap \mathcal L)(w)=\zero\) otherwise.
Thus, finance words outside the scenario language are assigned the semiring zero and do not contribute to the payoff evaluation.

To enforce such restrictions within a WFFA, interval membership must be encoded by
local payoff expressions. Let \(J\in\mathcal I(\DD)\). A local payoff expression \(g_J\in\EE_{\FF}\) is an \emph{interval guard} for \(J\) if, for every \(d\in\DD\), we have
\(\sem{g_J}(d)=\one\) when \(d\in J\), and \(\sem{g_J}(d)=\zero\) otherwise.
We call \(\FF\) \emph{interval-complete} if every interval \(J\in\mathcal I(\DD)\) has an interval guard \(g_J\in\EE_{\FF}\). For instance, the linear finance semiring \(\FF_{\lin}\) is interval-complete (cf. Example~2.9 in~\cite{DN26}).

\begin{theorem}[Restriction by event history automata]
\label{THM:wffa_restriction_eha}
Let \(\FF\) be interval-complete. For every EHA \(\H\) over \(\Sigma\) and
\(\DD\), and every WFFA \(\mathcal W\) over \(\Sigma\) and \(\FF\), there exists
a WFFA \(\mathcal W'\) over \(\Sigma\) and \(\FF\) such that
\(
\sem{\mathcal W'}=\sem{\mathcal W}\cap \mathcal L(\H),
\)
where \(\mathcal L(\H)\) denotes the scenario language induced by \(\H\).
\end{theorem}
The construction is a synchronized product of \(\H\) and \(\mathcal W\). The
interval-completeness assumption ensures that the EHA interval output can be
encoded as a local payoff-expression guard on the corresponding WFFA transitions.

\paragraph{Extremal Payoff Analysis}
\label{SUBSEC:extremal_values}
We now define the extremal payoff quantities used below.  We
restrict attention to the linear finance semiring \(\FF_{\lin}\), where payoffs
are real values and infeasible words have value \(-\infty\).

Let \(\F:(\Sigma\times\DD)^*\to \RR\cup\{-\infty\}\) be a quantitative finance
language. Its feasible domain is
\(
\dom(\F)=\{w\in(\Sigma\times\DD)^*\mid \F(w)\neq -\infty\}.
\)
Given a scenario language \(\mathcal L\subseteq(\Sigma\times\DD)^*\), the \emph{best-case} and
\emph{worst-case payoffs} over \(\mathcal L\) are
\(
\best_L(\F)=\sup \left\{\F(w) \mid w\in \mathcal L\cap\dom(\F)\right\}
\)
and
\(
\worst_L(\F)=\inf \left\{\F(w) \mid w\in \mathcal L\cap\dom(\F)\right\}
\),
respectively. If \({\mathcal L\cap\dom(\F)=\emptyset}\), then the scenario is infeasible.

Let \(\W=(Q,I,T,F,\wt_I,\wt_T,\wt_F)\) be a WFFA over \(\Sigma\) and \(\FF\).
A transition \(t\in T\) is \emph{enabled} by \(d\in\DD\) if
\(
\sem{\wt_T(t)}(d)\neq\zero.
\)
The WFFA \(\W\) is called \emph{deterministic} if it has a unique initial state and,
for every \(q\in Q\), \(a\in\Sigma\), and \(d\in\DD\), at most one transition
\(t=(q,a,q')\in T\) is enabled by \(d\).

The following result adapts graph-based extremal-value techniques for weighted
automata~\cite{ABK22} to WFFAs with data-dependent transition expressions and
scenario constraints, building on the threshold-analysis perspective
of~\cite{DN26}.

\begin{theorem}[Polynomial-time extremal payoff analysis]
\label{THM:extremal_payoff_analysis}
Let \(\H\) be an EHA over \(\Sigma\) and \(\DD\), and let \(\W\) be a WFFA over
\(\Sigma\) and \(\FF_{\lin}\). Assume that all numerical constants occurring in
\(\W\) and \(\H\) are rational numbers encoded in binary. Then the best-case value
\(\best_{\mathcal L(\H)}(\sem{\W})\) can be computed in polynomial time in the
size of the product automaton and the bit-size of the numerical constants. If
\(\W\) is deterministic, then the worst-case value
\(\worst_{\mathcal L(\H)}(\sem{\W})\) can also be computed in polynomial time in
the size of the product automaton and the bit-size of the numerical constants.
\end{theorem}

\section{Case Study: Autocallable Structured Product}

\emph{Autocallable structured products} are path-dependent instruments
with early redemption features; see, e.g.,~\cite{Gui15}.
We consider a simplified autocallable structured product with nominal value \(100\)
and horizon \(n\). 
At each observation date, the investor receives a coupon of \(2\), unless the
contract has already terminated. Two consecutive significant positive movements
trigger early redemption with nominal repayment and a \(10\%\) \emph{participation
premium} on the positive performance above the initial level \(100\). Two consecutive
significant negative movements deactivate protection; if maturity is reached
afterwards, the investor receives the terminal underlying value capped at \(100\).

In this case study, the payoff means the total undiscounted cash amount received by
the investor along a contract path, including coupons, redemption payments, and
possible participation premia. Discounting is omitted in order to keep the example
focused on path dependence and scenario constraints.

We model the payoff logic by a WFFA \(\mathcal P\) over the linear finance semiring
\(\FF_{\lin}\). The data domain is \(\DD=\RR_{\ge 0}\), and data values represent
observed underlying levels. The event alphabet is
\(
\Sigma=\{\texttt{f},\texttt{u},\texttt{d},\texttt{e},\texttt{a}\}.
\)
Here, \(\texttt{f}\), \(\texttt{u}\), and \(\texttt{d}\) denote flat movements,
significant upward movements, and significant downward movements of the
underlying price, respectively. The symbols
\(\texttt{e}\) and \(\texttt{a}\) denote regular expiration and autocall redemption.
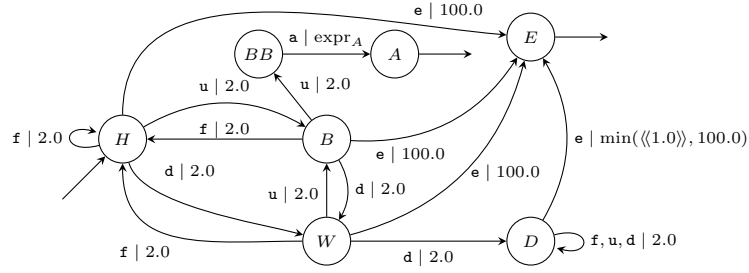
\begin{figure}[ht]
\centering
\begin{tikzpicture}[
    >=stealth,
    every node/.style={font=\scriptsize},
    state/.style={circle, draw, minimum size=20pt, inner sep=1pt},
    scale=0.9, transform shape
]
\node (hi) at (0.25,-1) {};
\node[state] (h) at (1.25,0) {$H$};
\node[state] (b) at (4.25,0) {$B$};
\node[state] (bb) at (3.25,1.25) {$BB$};
\node[state] (a) at (5.25,1.25) {$A$};
\node[state] (w) at (4.25,-1.5) {$W$};
\node[state] (d) at (7.25,-1.5) {$D$};
\node[state] (e) at (7.25, 1.5) {$E$};
\node (af) at (6.5, 1.25) {};
\node (ef) at (8.5, 1.5) {};

\draw[->]
    (hi) edge[above] (h)
    (h) edge[loop left] node{$\texttt{f} \mid 2.0$} (h)
    (h) edge[out=90, in=175, looseness=1.0, above] node[xshift=80pt, yshift=-6pt]{$\texttt{e} \mid 100.0$} (e)
    (h) edge[bend left, above] node{$\texttt{u} \mid 2.0$} (b)
    (h) edge[out=285,in=165,looseness=0.6] node[xshift=-6pt, yshift=16pt]{$\texttt{d} \mid 2.0$} (w)
    (b) edge[above] node[yshift=-3pt]{$\texttt{f} \mid 2.0$} (h)
    (b) edge[right] node[yshift=5pt]{$\texttt{u} \mid 2.0$} (bb)
    (b) edge[bend left] node[right]{$\texttt{d} \mid 2.0$} (w)
    (b) edge[bend right] node[yshift=-13pt, xshift=-15pt]{$\texttt{e} \mid 100.0$} (e)
    (bb) edge[above] node{$\texttt{a} \mid \mathrm{expr}_{A}$} (a)
    (w) edge[out=180,in=270,looseness=1.2] node[below left]{$\texttt{f} \mid 2.0$}(h)
    (w) edge[bend right] node[right]{$\texttt{e} \mid 100.0$} (e)
    (w) edge[above] node[left, yshift=-2pt]{$\texttt{u} \mid 2.0$} (b)
    (w) edge[below] node{$\texttt{d} \mid 2.0$} (d)
    (d) edge[loop right] node{$\texttt{f}, \texttt{u}, \texttt{d}\mid 2.0$} (d)
    (d) edge[bend right] node[right]{$\texttt{e} \mid \min(\llangle 1.0 \rrangle, 100.0)$} (e)
    (a) edge[above] (af)
    (e) edge[above] (ef);

\end{tikzpicture}
\caption{WFFA \(\mathcal{P}\) representing the payoff structure of the autocallable contract.}
\label{fig:structured_payoff}
\end{figure}
Figure~\ref{fig:structured_payoff} shows the WFFA \(\mathcal P\). Its states encode the relevant path-dependent status of the contract:
\(H\) (healthy) is the normal protected state, \(B\) (bullish) records one
significant positive movement, \(BB\) (bullish-bullish) records two consecutive
significant positive movements, \(W\) (weakening) records one significant negative
movement, and \(D\) (distressed) records protection loss after two consecutive
significant negative movements. The states \(A\) (autocall) and \(E\) (expired) are
terminal states for early redemption and regular expiration, respectively.
The initial state is \(H\), and the final states are \(A\) and \(E\). Their final
weights are \(0.0\) and are omitted in Figure~\ref{fig:structured_payoff}.
Non-terminal transitions carry the coupon contribution \(2.0\). The transition
\(BB\to A\) models autocall redemption and has weight 
\[
{\mathrm{expr}_{A}
=
100.0 \otimes \bigl((\bind{0.1}\otimes(-10.0))\oplus 0.0\bigr),
}\]
so that
\(
\sem{\mathrm{expr}_{A}}(d)
=
100.0+\max(0.1(d-100.0),0.0).
\)
Here, \(100.0\) represents the nominal repayment, and the second term is the
participation premium paid on positive performance above the initial level.
The transition \(D\to E\) models maturity after protection loss and pays
\(
\min(d,100.0).
\)
This payoff can be expressed by local payoff expressions over \(\FF_{\lin}\), using
the guard constructions from Example~2.9 in~\cite{DN26}. The automaton is
deterministic, since the next contract state is uniquely determined by the current
state and the observed event.

The automaton \(\mathcal P\) describes the payoff component independently of a
fixed contract duration. The finite contract structure, including the observation horizon and admissible termination events, is represented by a separate WFFA component and combined with \(\mathcal P\). We omit this auxiliary component from the presentation and focus on the payoff and scenario-constraint aspects relevant to the analysis.

\paragraph{Scenario configurations.}
We evaluate the payoff model under the three pattern-based scenario systems
summarized in Table~\ref{tab:scenario_systems}. Each system uses a priority
resolver with default interval \(J_0=[99,101]\). Thus, if no listed override
constraint is active, admissible prices are restricted to the default interval.

The listed constraints describe market-regime assumptions by assigning admissible
price intervals to event histories matching certain patterns. Priorities are given
by the order in which constraints are listed in Table~\ref{tab:scenario_systems}:
earlier constraints override later ones. For example, the constraint
\((\Sigma^{\ge 4}\cdot\texttt{u}\cdot\Sigma^*,\emptyset)\) in \(S_3\)
excludes histories in which an upward movement occurs after the fourth event.
Intuitively, this represents a scenario in which the underlying is not expected
to move upward in the later part of the analysis horizon.

\begin{table}[ht]
\centering
\scriptsize
\renewcommand{\arraystretch}{1.15}
\setlength{\tabcolsep}{3pt}
\begin{tabular}{c|l|m{7cm}}
Scenario & Constraints & Interpretation\\
\hline

\(S_1\) &
\begin{tabular}[c]{@{}l@{}}
\((\Sigma^*\cdot\texttt{uu}\cdot\Sigma^*,[180,220])\)\\
\((\Sigma^*\cdot\texttt{dd}\cdot\Sigma^*,[40,90])\)
\end{tabular}
&
High-volatility scenario: consecutive upward movements lead to high admissible
prices, while consecutive downward movements lead to stressed prices.
\\
\hline

\(S_2\) &
\begin{tabular}[c]{@{}l@{}}
\((\Sigma^*\cdot\texttt{uu}\cdot\Sigma^*,[105,119])\)\\
\((\Sigma^*\cdot\texttt{dd}\cdot\Sigma^*,[80,95])\)
\end{tabular}
&
Moderate market scenario: trend patterns are allowed, but price intervals remain
close to the nominal level.
\\
\hline

\(S_3\) &
\begin{tabular}[c]{@{}l@{}}
\((\Sigma^{\ge 4}\cdot\texttt{u}\cdot\Sigma^*,\emptyset)\)\\
\((\Sigma^*\cdot\texttt{uu}\cdot\Sigma^*,[140,180])\)\\
\((\Sigma^*\cdot\texttt{dd}\cdot\Sigma^*,[40,90])\)
\end{tabular}
&
Stress-oriented scenario: late upward movements are excluded, while downward
patterns remain associated with stressed prices.
\\

\end{tabular}
\vspace{0.2cm}
\caption{Scenario systems used in the evaluation.}
\label{tab:scenario_systems}
\end{table}
For each scenario system \(S_i\), \(i=1,2,3\), we first compile \(S_i\) into an
EHA \(\mathcal H_i\) using Theorem~\ref{THM:equiv_pattern_constraints_eha}. The full contract model is obtained by synchronizing the payoff WFFA \(\mathcal P\) with a fixed deterministic WFFA \(\mathcal T_n\) encoding the contract horizon and valid termination events. We then restrict this contract WFFA by \(\mathcal H_i\) using Theorem~\ref{THM:wffa_restriction_eha}.
The resulting WFFA represents exactly the executions that satisfy both the
contract structure and the scenario constraints, while preserving the payoff
accumulated by \(\mathcal P\). Best-case and worst-case payoffs, together with
witness event histories, are computed using
Theorem~\ref{THM:extremal_payoff_analysis}.

\subsection{Experimental Setup}

We implemented a Java research prototype that supports an end-to-end workflow from
declarative model specification to extremal payoff analysis. Contract models and
scenario assumptions are specified by JSON task descriptions; pattern-based
scenario constraints are then compiled into EHAs.
The resulting components are combined using synchronized product constructions.
The final scenario-constrained WFFA is then analyzed to compute exact best-case
and worst-case payoffs together with witness event histories.

The evaluation considers two aspects. First, we compare the extremal payoffs
obtained under the scenario systems from Table~\ref{tab:scenario_systems}. Second,
we measure scalability with respect to the contract horizon and the number of
scenario constraints. For each configuration, we record the size of the constructed
automata and the end-to-end runtime.

All experiments were run with a Java prototype on OpenJDK~21.0.1 using Gradle on an
Apple MacBook Pro with an Apple M1 Max processor and 64 GB of memory. No custom JVM
tuning parameters were used. Reported
runtimes are averages over 10 independent executions and include parsing, automata
construction, product construction, extremal payoff computation, and witness
extraction. The prototype and experiment files are publicly available in the accompanying
repository.\footnote{\url{https://gitlab.com/vitnberg/wffa}}

\section{Evaluation Results}

\paragraph{Scenario-Constrained Payoff Analysis.}

We first evaluate exact extremal payoffs under the scenario systems from
Table~\ref{tab:scenario_systems}. For each configuration, the framework computes
best-case and worst-case payoffs together with witness event histories; the
results are shown in Table~\ref{tab:extremal_analysis_results}.
\begin{table}[ht]
\centering
\scriptsize
\setlength{\tabcolsep}{4pt}
\begin{tabular}{c|c|c|c|l|l}
Scenario & Horizon & Best & Worst
& Best witness & Worst witness\\
\hline
\(S_1\) & \(8\)  & \(126.0\) & \(56.0\)  & late autocall & protection-loss expiration\\
\(S_1\) & \(64\) & \(238.0\) & \(112.0\) & very late autocall & early autocall\\
\hline
\(S_2\) & \(8\)  & \(116.0\) & \(96.0\)  & regular expiration & mild distressed expiration\\
\(S_2\) & \(64\) & \(228.0\) & \(104.5\) & long coupon accumulation & early autocall\\
\hline
\(S_3\) & \(8\)  & \(116.0\) & \(56.0\)  & early autocall & protection-loss expiration\\
\(S_3\) & \(64\) & \(228.0\) & \(108.0\) & no late upward moves & early autocall\\
\end{tabular}
\vspace{0.2cm}
\caption{Extremal payoff results under scenario constraints.}
\label{tab:extremal_analysis_results}
\end{table}

Table~\ref{tab:extremal_analysis_results} shows that the same payoff automaton can
exhibit qualitatively different extremal behavior under different scenario
constraints. In \(S_1\), best-case
payoffs are obtained by delaying autocall redemption and accumulating coupons. For
long horizons, however, early autocall may become worst-case, because premature
termination prevents further coupon accumulation. This illustrates that a locally
favorable event may be globally adverse once path-dependent cash-flow accumulation
is taken into account.
Scenario \(S_2\) yields narrower payoff ranges, reflecting its moderate price
intervals. Scenario \(S_3\) excludes upward movements after the fourth event. Consequently,
for longer horizons the best-case payoff must be realized before this restriction
becomes binding, or through coupon accumulation along paths that remain admissible
without late upward movements. In all cases, the framework returns not only exact extremal values but
also witness histories explaining how these values are realized.

\paragraph{Scalability results.}
We study scalability along two dimensions: the contract horizon \(n\) and the
number \(|\mathbb C|\) of scenario constraints.
First, we vary the contract horizon \(n\) using scenario system \(S_1\). Figure~\ref{fig:scalability_horizon} reports the size of
the constructed product automaton and the corresponding end-to-end runtime.
The results show an approximately linear growth in the number of product states and
transitions as the horizon increases. Runtime grows moderately and remains below
seven seconds even for \(n=512\). Thus, for this scenario system, the exact
scenario-constrained extremal analysis remains tractable for long contract
horizons.
\begin{figure}[ht]
\centering
\begin{minipage}{0.45\textwidth}
\centering
\scriptsize
\setlength{\tabcolsep}{3pt}
\begin{tabular}{c|r|r|r}
$n$ & States & Transitions & Runtime (s)\\
\hline
4   & 1029  & 1134   & 0.061\\
8   & 1617  & 2226   & 0.093\\
16  & 2793  & 4410   & 0.148\\
32  & 5145  & 8778   & 0.273\\
64  & 9849  & 17514  & 0.609\\
128 & 19257 & 34986  & 1.280\\
256 & 38073 & 69930  & 2.715\\
512 & 75705 & 139818 & 6.628\\
\end{tabular}
\end{minipage}
\hfill
\begin{minipage}{0.5\textwidth}
\centering
\begin{tikzpicture}
\begin{axis}[
    width=\textwidth,
    height=0.62\textwidth,
    ylabel={Runtime (s)},
    grid=major,
    tick label style={font=\scriptsize},
    label style={font=\scriptsize},
    xlabel={},
    extra description/.code={
        \node[font=\scriptsize, anchor=north east]
        at (axis description cs:1.05,-0.01) {$n$};
    },
]

\addplot [
    black,
    very thick,
    mark=*,
    mark options={fill=black}
]coordinates {
    (4,0.061)
    (8,0.093)
    (16,0.148)
    (32,0.273)
    (64,0.609)
    (128,1.280)
    (256,2.715)
    (512,6.628)
};
\end{axis}
\end{tikzpicture}
\end{minipage}
\caption{Scalability with respect to the contract horizon for scenario system \(S_1\).}
\label{fig:scalability_horizon}
\end{figure}
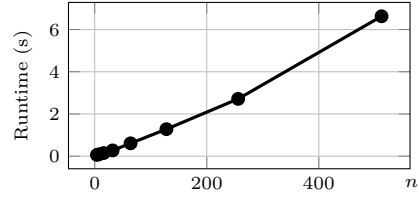

Second, we vary the number \(|\mathbb C|\) of scenario constraints while keeping the
contract horizon fixed at \(n=8\). Figure~\ref{fig:scalability_constraints} reports the
resulting product-automaton sizes and end-to-end runtimes.
In this experiment, increasing the number of constraints has a stronger effect
than increasing the horizon. Product sizes and runtimes grow noticeably for larger
constraint sets, reflecting synchronization effects between overlapping event
patterns. At the same time, exact extremal analysis remains practical for the
considered range of scenario systems, despite the observed product-state growth.

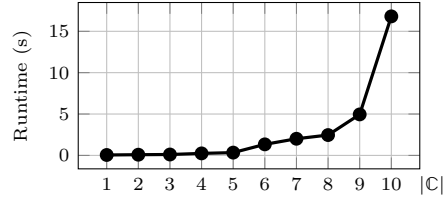
\begin{figure}[ht]
\centering
\begin{minipage}{0.45\textwidth}
\centering
\scriptsize
\setlength{\tabcolsep}{3pt}
\vspace{-0.3cm}
\begin{tabular}{c|r|r|r}
\(|\mathbb C|\) & States & Transitions & Runtime (s)\\
\hline
1  & 539    & 742    & 0.050\\
2  & 1617   & 2226   & 0.084\\
3  & 2002   & 2756   & 0.108\\
4  & 6468   & 8904   & 0.242\\
5  & 9163   & 12614  & 0.337\\
6  & 41657  & 57346  & 1.334\\
7  & 62601  & 86178  & 2.011\\
8  & 77847  & 107166 & 2.458\\
9  & 105875 & 145750 & 4.954\\
10 & 510664 & 702992 & 16.805\\
\end{tabular}
\end{minipage}
\hfill
\begin{minipage}{0.5\textwidth}
\centering
\begin{tikzpicture}
\begin{axis}[
    width=\textwidth,
    height=0.62\textwidth,
    ylabel={Runtime (s)},
    grid=major,
    tick label style={font=\scriptsize},
    label style={font=\scriptsize},
    xtick={1,2,3,4,5,6,7,8,9,10},
    xlabel={},
    extra description/.code={
        \node[font=\scriptsize, anchor=north east]
        at (axis description cs:1.1,0.01) {\(|\mathbb C|\)};
    },
]

\addplot [
    black,
    very thick,
    mark=*,
    mark options={fill=black}
] coordinates {
    (1,0.050)
    (2,0.084)
    (3,0.108)
    (4,0.242)
    (5,0.337)
    (6,1.334)
    (7,2.011)
    (8,2.458)
    (9,4.954)
    (10,16.805)
};
\end{axis}
\end{tikzpicture}
\end{minipage}
\caption{Scalability with respect to the number of scenario constraints for fixed contract horizon \(n=8\).}
\label{fig:scalability_constraints}
\end{figure}

\section{Conclusion and Future Work}

This paper introduced an automata-based framework for exact extremal payoff
analysis under declarative scenario constraints. The central idea is to separate
the deterministic contract model from assumptions about future market evolution:
payoffs are represented by WFFAs, while admissible scenarios are specified by
event history automata or compiled from pattern-based constraint systems.

The resulting product construction enables best-case and worst-case payoff
analysis over complete languages of admissible event histories, rather than over
sampled paths. In addition to extremal values, the framework produces witness
histories that make the corresponding outcomes interpretable.

From a computational perspective, the scalability experiments quantify the cost of
the synchronized product constructions underlying the framework. As expected,
product sizes and runtimes increase when horizons grow or when multiple scenario
constraints interact. However, the prototype still computes exact extremal values
and witness histories across all tested configurations, showing that the approach
is practically tractable for moderately large horizons and nontrivial scenario
systems.

Future work can proceed in several directions. First, the current EHA model assigns
admissible intervals from finite-state event histories. Many realistic market
assumptions, however, involve quantitative dependencies between observations, such
as bounded relative changes, reversals, drawdown constraints, or dynamically
evolving reference levels. Extending the framework toward multi-valued finance
words and models such as weighted register automata~\cite{SYT21} may provide a way to capture
such dependencies while preserving algorithmic analyzability.

Second, the synchronized product constructions used in this paper are explicit.
An important algorithmic direction is the development of lazy or on-the-fly
product constructions that generate only those combined states needed during
extremal payoff analysis. This could improve scalability for scenario systems
with many interacting constraints.

Finally, the same methodology could be applied beyond structured products. One
promising direction is the analysis of rule-based investment or trading strategies,
where technical-analysis patterns~\cite{Mur99} and event-driven signals can be represented as
formal scenario constraints and evaluated against quantitative payoff objectives.

\begin{credits}
\subsubsection{\discintname}
The author has no competing interests to declare that are relevant to the content of this article.
\end{credits}

\bibliographystyle{splncs04}
\bibliography{eha}

\appendix

\section{Proofs}

\subsection{Proof of Theorem~\ref{THM:equiv_pattern_constraints_eha}}

Recall that a deterministic finite automaton over \(\Sigma\) is a tuple
\(\mathcal A=(Q,q_0,\delta,F)\), where \(Q\) is a finite set of states,
\(q_0\in Q\) is the initial state, \(\delta:Q\times\Sigma\to Q\) is the transition
function, and \(F\subseteq Q\) is the set of final states.

A word \(u\in\Sigma^*\) is accepted by \(\mathcal A\) if the unique run of
\(\mathcal A\) on \(u\) ends in a final state. The language recognized by
\(\mathcal A\), denoted \(L(\mathcal A)\), is the set of all words accepted by
\(\mathcal A\).

We use Kleene's theorem and the standard determinization construction
(cf.~\cite{Kle56,HU79}): every regular expression \(E\in\Reg(\Sigma)\) can be translated into a deterministic finite automaton \(\mathcal A_E\) such that
\(L(\mathcal A_E)=\sem{E}\).

\paragraph{\textbf{Theorem~\ref{THM:equiv_pattern_constraints_eha}.}}
{\itshape
Let \(\Sigma\) be a finite event alphabet and let \(\DD\) be the quantitative domain.

\begin{enumerate}[label=(\alph*)]
\item For every pattern-based scenario constraint system \(\mathcal S\) over \(\Sigma\) and \(\DD\), there exists an EHA \(\H\) over \(\Sigma\) and \(\DD\) such that
\(
\sem{\H}=\sem{\mathcal S}.
\)

\item For every EHA \(\H\) over \(\Sigma\) and \(\DD\), there exists a pattern-based scenario constraint system \(\mathcal S\) over \(\Sigma\) and \(\DD\) such that
\(
\sem{\mathcal S}=\sem{\H}.
\)
\end{enumerate}
}

\begin{proof}
\begin{enumerate}[label=(\alph*)]
\item Let \(\mathcal S=(\mathbb C,\operatorname{Res})\) be a pattern-based scenario
constraint system over \(\Sigma\) and \(\DD\), with
\(\mathbb C=(c_1,\dots,c_m)\) and \(c_i=(E_i,J_i)\) for \(i\in\{1, \dots, m\}\).

We construct an event history automaton \(\H\) in two steps.

\begin{enumerate}
\item[(i)] For each \(i\in\{1,\dots,m\}\), choose a complete deterministic finite
automaton
\(
\mathcal D_i=(Q_i,q_{0,i},\delta_i,F_i)
\)
over \(\Sigma\) such that
\(
L(\mathcal D_i)=\sem{E_i}.
\)

\item[(ii)] Define the EHA
\(
\H=(Q,q_0,\delta,\rho)
\)
as the synchronous product of the automata \(\mathcal D_1,\dots,\mathcal D_m\).
Let
\(
Q=Q_1\times\cdots\times Q_m
\)
and
\(
q_0=(q_{0,1},\dots,q_{0,m}).
\)
The transition function \(\delta:Q\times\Sigma\to Q\) is defined by
\[
\delta((q_1,\dots,q_m),a)
=
(\delta_1(q_1,a),\dots,\delta_m(q_m,a))
\]
for all \((q_1,\dots,q_m)\in Q\) and \(a\in\Sigma\). Finally, define the output
function \(\rho:Q\to\mathcal I(\DD)\) by
\[
\rho(q_1,\dots,q_m)
=
\operatorname{Res}\bigl(\{i\in\{1, \dots, m\}\mid q_i\in F_i\}\bigr).
\]
\end{enumerate}
It remains to show that \(\sem{\H}=\sem{\mathcal S}\). Let \(u\in\Sigma^*\).
By construction of the product automaton, the \(i\)-th component of
\(\operatorname{last}_{\H}(u)\) is precisely
\(\operatorname{last}_{\mathcal D_i}(u)\).
Hence,
\(
\operatorname{last}_{\mathcal D_i}(u)\in F_i
\iff
u\in \sem{E_i}.
\)
Therefore,
\[
\{i\in\{1, \dots, m\}\mid \operatorname{last}_{\mathcal D_i}(u)\in F_i\}
=
\operatorname{Act}_{\mathcal S}(u).
\]
It follows that
\[
\sem{\H}(u)
=
\rho(\operatorname{last}_{\H}(u))
=
\operatorname{Res}(\operatorname{Act}_{\mathcal S}(u))
=
\sem{\mathcal S}(u).
\]
Since this holds for all \(u\in\Sigma^*\), we obtain
\(
\sem{\H}=\sem{\mathcal S}.
\)

\item For the converse direction, let \(\H=(Q,q_0,\delta,\rho)\) be an EHA over
\(\Sigma\) and \(\DD\). Without loss of generality, assume that \(Q=\{1,\dots,m\}\).
For each \(j\in\{1,\dots,m\}\), let
\[
L_j=\{u\in\Sigma^*\mid \operatorname{last}_{\H}(u)=j\}.
\]
This language is recognized by the deterministic finite automaton
\((Q,q_0,\delta,\{j\})\), obtained from \(\H\) by taking \(j\) as the only final
state. Hence, choose a regular expression \(E_j\in\Reg(\Sigma)\) such that
\(\sem{E_j}=L_j\).

Define \(\mathcal S=(\mathbb C,\operatorname{Res})\) by taking
\(\mathbb C=(c_1,\dots,c_m)\), where \(c_j=(E_j,\rho(j))\), and letting
\(\operatorname{Res}=\operatorname{Res}_{\mathrm{prio}}^{J_0}\) for an arbitrary
interval \(J_0\in\mathcal I(\DD)\).

For every \(u\in\Sigma^*\), exactly one constraint is active, namely the one
corresponding to \(\operatorname{last}_{\H}(u)\). Hence, if
\(\operatorname{last}_{\H}(u)=j\), then
\(
\operatorname{Act}_{\mathcal S}(u)=\{j\}
\).
Since the active set is a singleton, the priority resolver returns
\(
\operatorname{Res}(\{j\})=\rho(j)
\).
Therefore,
\[
\sem{\mathcal S}(u)
=
\operatorname{Res}(\operatorname{Act}_{\mathcal S}(u))
=
\rho(j)
=
\sem{\H}(u).
\]
Since this holds for all \(u\in\Sigma^*\), we obtain
\(\sem{\mathcal S}=\sem{\H}\).
\qed
\end{enumerate}
\end{proof}

\subsection{Proof of Theorem~\ref{THM:wffa_restriction_eha}}

\paragraph{\textbf{Theorem~\ref{THM:wffa_restriction_eha}}.}
{\itshape
Let \(\FF\) be interval-complete. For every EHA \(\H\) over \(\Sigma\) and
\(\DD\), and every WFFA \(\mathcal W\) over \(\Sigma\) and \(\FF\), there exists
a WFFA \(\mathcal W'\) over \(\Sigma\) and \(\FF\) such that
\(
\sem{\mathcal W'}=\sem{\mathcal W}\cap \mathcal L(\H),
\)
where \(\mathcal L(\H)\) denotes the scenario language induced by \(\H\).
}

\begin{proof}
Let:
\begin{itemize}
\item \(\H=(P,p_0,\delta,\rho)\) be an EHA over \(\Sigma\) and \(\DD\);
\item \(\W=(Q,I,T,F,\wt_I,\wt_T,\wt_F)\) be a WFFA over \(\Sigma\) and \(\FF\).
\end{itemize}
Since \(\FF\) is interval-complete, for every state \(p\in P\) there exists a
local payoff expression \(g_p\in\EE_{\FF}\) representing membership in the
interval \(\rho(p)\in\mathcal I(\DD)\), i.e., for all \(d\in\DD\),
\[
\sem{g_p}(d)=
\begin{cases}
\one, & d\in\rho(p),\\
\zero, & d\notin\rho(p).
\end{cases}
\]

We define \(\W'\) as the product automaton whose states record both the current
state of \(\H\) and the current state of \(\W\). The EHA component determines the
interval guard to be applied after reading the next event. Formally, let
\(\W'=(Q',I',T',F',\wt'_I,\wt'_T,\wt'_F)\), where:
\begin{itemize}
\item \(Q'=P\times Q\), \(I'=\{p_0\} \times I\), \(F'=P\times F\);
\item \(T'\) consists of all transitions
\(
((p,q),a,(p',q'))
\)
such that \(p'=\delta(p,a)\) and \((q,a,q')\in T\);
\item \(\wt'_I(p_0,q)=\wt_I(q)\) for all \(q\in I\);
\item \(\wt'_F(p,q)=\wt_F(q)\) for all \(p\in P\) and \(q\in F\);
\item for \(t'=((p,q),a,(p',q'))\in T'\), with \(t=(q,a,q')\in T\), set
\[
\wt'_T(t')=\wt_T(t)\otimes g_{p'}.
\]
Here \(\wt_T(t),g_{p'}\in\EE_{\FF}\), so the expression is well-defined by closure
of \(\EE_{\FF}\) under \(\otimes\).
\end{itemize}

Let \(w=(a_1,d_1)\dots(a_n,d_n)\). The EHA has a unique run
\(p_0,p_1,\dots,p_n\) on \(a_1\dots a_n\), where
\(p_i=\delta(p_{i-1},a_i)\). Runs of \(\W'\) on \(w\) are in one-to-one
correspondence with runs of \(\W\) on \(w\): a run
\(q_0,\dots,q_n\) of \(\W\) corresponds to
\((p_0,q_0),\dots,(p_n,q_n)\) in \(\W'\).

If \(w\in\mathcal L(\H)\), then \(d_i\in\rho(p_i)\) for every \(i\). Hence
\(\sem{g_{p_i}}(d_i)=\one\), and every corresponding run in \(\W'\) has the same
weight as the original run in \(\W\). Therefore \(\sem{\W'}(w)=\sem{\W}(w)\).

If \(w\notin\mathcal L(\H)\), then for some \(i\) we have
\(d_i\notin\rho(p_i)\). Hence \(\sem{g_{p_i}}(d_i)=\zero\), so every run of
\(\W'\) on \(w\) has weight \(\zero\). Thus \(\sem{\W'}(w)=\zero\).

Consequently,
\[
\sem{\W'}(w)=
\begin{cases}
\sem{\W}(w), & w\in\mathcal L(\H),\\
\zero, & \text{otherwise,}
\end{cases}
\]
for all \(w\in(\Sigma\times\DD)^*\). Hence
\(\sem{\W'}=\sem{\W}\cap\mathcal L(\H)\).
\qed
\end{proof}

\subsection{Proof of Theorem~\ref{THM:extremal_payoff_analysis}}

\paragraph{\textbf{Theorem~\ref{THM:extremal_payoff_analysis}}.}
{\itshape
Let \(\H\) be an EHA over \(\Sigma\) and \(\DD\), and let \(\W\) be a WFFA over
\(\Sigma\) and \(\FF_{\lin}\). Assume that all numerical constants occurring in
\(\W\) and \(\H\) are rational numbers encoded in binary. Then the best-case value
\(\best_{\mathcal L(\H)}(\sem{\W})\) can be computed in polynomial time in the
size of the product automaton and the bit-size of the numerical constants. If
\(\W\) is deterministic, then the worst-case value
\(\worst_{\mathcal L(\H)}(\sem{\W})\) can also be computed in polynomial time in
the size of the product automaton and the bit-size of the numerical constants.
}

\begin{proof}
We give the argument for the best-case value; the worst-case value is analogous
under determinism. By Theorem~\ref{THM:wffa_restriction_eha}, we can construct a WFFA \(\W'\) such
that
\(
\sem{\W'}=\sem{\W}\cap\mathcal L(\H).
\)
Therefore,
\[
\best_{\mathcal L(\H)}(\sem{\W})
=
\sup_{w\in\mathcal L(\H)\cap\dom(\sem{\W})}\sem{\W}(w)
=
\sup_{w\in\dom(\sem{\W'})}\sem{\W'}(w).
\]
Thus, it suffices to compute the unrestricted supremum of the
scenario-constrained WFFA $\W'$. In the following, write this WFFA as
\(
\W'=(Q,I,T,F,\wt_I,\wt_T,\wt_F).
\)
For the linear finance semiring, Lemma~9.7 in~\cite{DN26} gives an effective
procedure for computing
\(
\sup_{d\in\DD}\sem{\wt_T(t)}(d)
\)
for each transition expression \(\wt_T(t)\).
For each transition \(t\in T\), compute
\(
M_t=\sup_{d\in\DD}\sem{\wt_T(t)}(d).
\)
If \(M_t=-\infty\), then \(t\) is infeasible and can be removed. If
\(M_t=+\infty\), then the best-case value is \(+\infty\) whenever \(t\) lies on
some feasible initial-to-final path. This condition is checked by standard graph reachability after removing infeasible transitions.

It remains to consider transitions with finite \(M_t\). Replacing each such
transition expression by the constant weight \(M_t\) yields a classical max-plus
weighted automaton \(\W_{\max}\) over the feasible finite-weight part of the same
underlying graph. Its supremum value can then be computed by the standard
graph-based analysis for weighted automata, as in~\cite{ABK22}.

After removing infeasible transitions, \(\dom(\sem{\W})=\emptyset\) iff
\(\dom(\sem{\W_{\max}})=\emptyset\). If the domains are empty, the scenario is
infeasible. Otherwise, let
\[
S=\sup\{\sem{\W}(w)\mid w \in (\Sigma \times \DD)^* \text{ and } \sem{\W}(w) \neq -\infty\}
\]
and
\[
S_{\max}=\sup\{\sem{\W_{\max}}(u)\mid u\in\Sigma^*
\text{ and } \sem{\W_{\max}}(u) \neq -\infty\}.
\]
We show that \(S=S_{\max}\).

First, \(S\le S_{\max}\). Indeed, every accepting run of \(\W\) on a finance word
\(w=(a_1,d_1)\dots(a_n,d_n)\) determines an accepting path
\(\pi=t_1\dots t_n\) in \(\W_{\max}\). Since
\(\sem{\wt_T(t_i)}(d_i)\le M_{t_i}\) for all \(i \in \{1, \dots, n\}\) and \(d_i \in \DD\), the payoff of this run is at
most the weight of \(\pi\) in \(\W_{\max}\).

Conversely, let \(\pi=t_1\dots t_n\) be an accepting path in \(\W_{\max}\) with
finite weight and let \(\eta>0\). We show that there is a finance word whose
payoff is greater than \(\operatorname{wt}_{\W_{\max}}(\pi)-\eta\). If \(n=0\),
then the empty finance word realizes the same initial-final weight, so the claim
is immediate. Assume \(n>0\). Write \(t_i=(q_{i-1},a_i,q_i)\). By the definition of
\(M_{t_i}\), for each \(i\) there exists \(d^{\eta}_i\in\DD\) such that
\[
\sem{\wt_T(t_i)}(d^{\eta}_i)>M_{t_i}-\frac{\eta}{n} .
\]
Here the choices of \(d^{\eta}_i\) are independent: once the event path
\(a_1\dots a_n\) is fixed, automaton states and scenario guards are fixed by the
event history, and each transition expression depends only on the current value
\(d^{\eta}_i\). Hence the finance word
\(
w=(a_1,d^{\eta}_1)\dots(a_n,d^{\eta}_n)
\)
admits the corresponding run in \(\W\), whose payoff is greater than the weight
of \(\pi\) minus \(\eta\). Therefore
\(S\ge \operatorname{wt}_{\W_{\max}}(\pi)-\eta\). Since \(\pi\) and
\(\eta\) were arbitrary, \(S\ge S_{\max}\).

Thus \(S=S_{\max}\).

The value \(S_{\max}\) is computed by standard graph analysis for max-plus
automata. We first restrict the graph to states and transitions that occur on some
path from an initial state to a final state. If this accepting part of the graph
contains a cycle of positive total weight, then \(S_{\max}=+\infty\). Otherwise,
the finite value is computed in polynomial time by standard path-optimization
algorithms, such as a Bellman--Ford variant on the remaining weighted graph.

For the worst-case value, assume that \(\W\) is deterministic. Then every finance
word induces at most one accepting run, so the value \(\sem{\W}(w)\) is not affected
by maximization over competing runs. We use the identity
\[
\worst(\sem{\W})
=
-\sup_{w\in\dom(\sem{\W})}(-\sem{\W}(w)).
\]
It remains to observe that, for the linear finance semiring, transition
expressions can be effectively negated. By Lemma~6.23 in~\cite{DN26}, the
semantics of local payoff expressions can be reduced to finitely many affine
pieces. Negating each affine piece yields an effective representation of the
negative transition contribution. Hence we can construct a deterministic WFFA
\(-\W\) such that, for every \(w\in\dom(\sem{\W})\),
\[
\sem{-\W}(w)=-\sem{\W}(w).
\]

Applying the best-case procedure to \(-\W\) computes
\[
\sup_{w\in\dom(\sem{\W})}(-\sem{\W}(w)).
\]
Taking the negative of this value yields \(\worst(\sem{\W})\). Thus the worst-case
value is computable by the same polynomial-time graph analysis, after the
transition-level negation and extrema have been computed. \qed
\end{proof}

\section{Additional Case Study Details}

\subsection{Contract Temporal Structure}

The payoff component \(\mathcal P\) described in the main text specifies how cash
flows evolve in response to contract events, but it does not by itself fix the
temporal structure of the contract. In the case study, this structure is represented
by a separate deterministic automaton \(\mathcal T_n\) over the same event alphabet
\(\Sigma\). Its role is to constrain admissible event sequences to those consistent
with an \(n\)-period contract with observation dates, possible autocall redemption,
and terminal expiration.

The states \(\tau_0,\dots,\tau_n\) represent successive observation periods, where
\(\tau_i\) is reached after \(i\) observation events. The state \(A\) represents
early termination by autocall, and \(E\) represents regular expiration. The initial
state is \(\tau_0\), and the final states are \(A\) and \(E\). Transitions labeled
\(\texttt{f}\), \(\texttt{u}\), and \(\texttt{d}\) advance the contract to the next
observation period, while \(\texttt{a}\) and \(\texttt{e}\) model autocall
redemption and final expiration, respectively.

The automaton \(\mathcal T_n\) is shown in Fig.~\ref{fig:temporal_structure}.
For composition with the payoff component, we regard \(\mathcal T_n\) as a WFFA
whose initial, transition, and final weights are all the semiring unit \(\one\)
(\(0.0\) in \(\FF_{\lin}\)). The complete contract model is then obtained as the
Hadamard product of this structural component with \(\mathcal P\)
(cf.~\cite{DN26}, Lemma~5.4). Thus, \(\mathcal P\) determines the payoff
accumulation, whereas \(\mathcal T_n\) restricts the event sequences to those that
are temporally admissible.

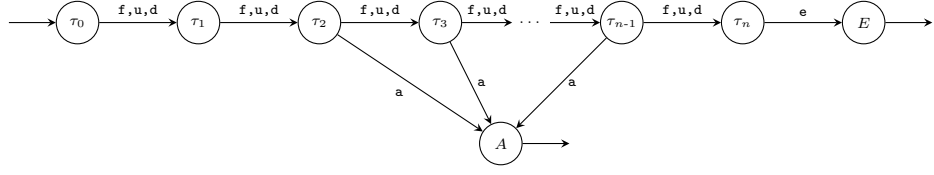
\begin{figure}[ht]
\centering
\begin{tikzpicture}[
    >=stealth,
    every node/.style={font=\scriptsize},
    state/.style={circle, draw, minimum size=20pt, inner sep=1pt},
    scale=0.8, transform shape
]
\node (s0i) at (0,0) {};
\node[state] (s0) at (1.25,0) {$\tau_0$};
\node[state] (s1) at (3.25,0) {$\tau_1$};
\node[state] (s2) at (5.25,0) {$\tau_2$};
\node[state] (s3) at (7.25,0) {$\tau_3$};
\node[state] (snminusone) at (10.25,0) {$\tau_{n\text{-}1}$};
\node[state] (sn) at (12.25,0) {$\tau_{n}$};
\node[state] (e) at (14.25,0) {$E$};
\node[state] (a) at (8.25,-2) {$A$};
\node (dots) at (8.75,0) {$\dots$};
\node (ef) at (15.5,0) {};
\node (af) at (9.5,-2.0) {};

\draw[->]
    (s0i) edge[above] (s0)
    (s0) edge[above] node{$\texttt{f,u,d}$} (s1)
    (s1) edge[above] node{$\texttt{f,u,d}$} (s2)
    (s2) edge[above] node{$\texttt{f,u,d}$} (s3)
    (s3) edge[above] node{$\texttt{f,u,d}$} (dots)
    (dots) edge[above] node{$\texttt{f,u,d}$} (snminusone)
    (snminusone) edge[above] node{$\texttt{f,u,d}$} (sn)
    (sn) edge[above] node{$\texttt{e}$} (e)
    (s2) edge node[below left]{$\texttt{a}$} (a)
    (s3) edge node[right]{$\texttt{a}$} (a)
    (snminusone) edge node[right]{$\texttt{a}$} (a)
    (e) edge[above] (ef)
    (a) edge[above] (af)
    ;

\end{tikzpicture}
\caption{Deterministic automaton \(\mathcal T_n\) constraining the temporal
structure of an \(n\)-period contract.}
\label{fig:temporal_structure}
\end{figure}

\section{Prototype Architecture and Reproducibility}

\paragraph{Prototype Architecture.}

We implemented a Java research prototype supporting the end-to-end workflow used
in the case study. The prototype separates declarative model specification,
automata construction, scenario compilation, product construction, and extremal
payoff analysis.

Input models are provided as JSON task descriptions. A task specifies the payoff
WFFA, the temporal contract structure, the pattern-based scenario constraint
system, and the analysis to be performed. The workflow is as follows. The JSON
specification is deserialized into internal DTOs and translated into automata-level
objects. The payoff component \(\mathcal P\) is synchronized with the temporal
structure \(\mathcal T_n\). The scenario system \(\mathcal S\) is compiled into an
EHA \(\mathcal H_{\mathcal S}\), which is then synchronized with the contract
model. The resulting scenario-constrained WFFA is analyzed to compute exact
best-case and worst-case payoffs, witness event histories, automaton sizes, and
runtime statistics.

The implementation consists of four main layers: an automata layer for deterministic
automata, WFFAs, and product constructions; a scenario layer for pattern-based
constraints and EHAs; a JSON layer for declarative input specifications; and a
command-line layer for running analysis tasks. This modular organization allows
payoff models, temporal structures, and scenario assumptions to be varied
independently and recombined through the workflow shown in
Fig.~\ref{fig:scenario_constrained_workflow}.

\begin{figure}[t]
\centering
\begin{tikzpicture}[
    >=stealth,
    node distance=0.3cm and 1.0cm,
    box/.style={
        draw,
        align=center,
        text width=2cm,
        minimum width=2.0cm,
        minimum height=0.5cm,
        font=\fontsize{6}{7}\sffamily
    },
    arrow/.style={->}
]

\node[box] (P) {Payoff WFFA\\ \(\mathcal{P}\)};
\node[box, below=of P] (T) {Temporal structure \(\mathcal{T}_n\)};
\node[box, below=of T] (S) {Scenario constraint\\ system \(\mathcal{S}\)};

\node[box, right=of T, yshift=0.4cm] (PT) {Contract WFFA \(\mathcal{W}_{\mathcal{P},\mathcal{T}_n}\)};

\node[box, right=of S] (H) {EHA \(\mathcal{H}_{\mathcal{S}}\)};

\node[box, right=of PT, yshift=-0.8cm] (Final) {Scenario-constrained WFFA
\(\mathcal W\)};

\node[box, right=of Final] (Analysis) {Best-/worst-case analysis on WFFA \(\W\)};

\draw[arrow] (P.east) -- (PT.west);
\draw[arrow] (T.east) -- (PT.west);

\draw[arrow] (S.east) -- (H.west);

\draw[arrow] (PT.east) -- (Final.west);
\draw[arrow] (H.east) -- (Final.west);

\draw[arrow] (Final.east) -- (Analysis.west);

\end{tikzpicture}
\caption{Prototype workflow for scenario-constrained payoff analysis.}
\label{fig:scenario_constrained_workflow}
\end{figure}
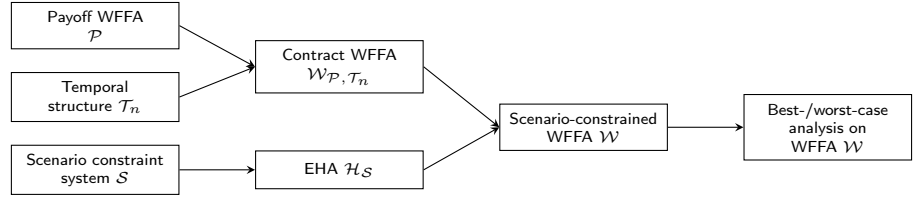

\paragraph{Input format.}

Analysis tasks are specified as JSON documents. The example below shows a
pattern-based scenario analysis task. The field \texttt{wffaComposition} describes
the contract model as a composition of registered WFFA components. In the case
study, the payoff structure and the temporal contract structure are stored in a
model registry and referenced by identifiers. This follows the compositional view
of WFFAs used in~\cite{DN26}: complex contract models are assembled from smaller
automata components by product constructions.

The field \texttt{scenario} specifies the event alphabet, the resolver, and the
list of pattern-based interval constraints. In the example, the resolver is a
priority resolver with fallback interval \([99,101]\). The listed constraint
matches histories containing two consecutive upward movements and assigns the
admissible interval \([180,220]\) whenever this pattern is active.

\begin{lstlisting}[language=json,basicstyle=\ttfamily\scriptsize, caption={Example JSON task description.},label={lst:json_task}]
{
  "type": "pattern_scenario_analysis",
  "wffaComposition": {
    "type": "combine",
    "factors": [
      { "type": "ref", "referenceId": "autocallable_payoff_structure" },
      { "type": "ref", "referenceId": "autocallable_temporal_structure_t8" }
    ]
  },
  "scenario": {
    "alphabet": ["FLAT", "UP", "DOWN", "AUTOCALL", "EXPIRATION"],
    "resolver": {
      "type": "priority",
      "fallbackInterval": {
        "type": "between",
        "lowerBound": 99.0,
        "upperBound": 101.0
      }
    },
    "constraints": [
      {
        "pattern": {
          "type": "concatenation",
          "expressions": [
            { "type": "any_word"},
            { "type": "literal_sequence", "sequence": ["UP", "UP"] },
            { "type": "any_word" }
          ]
        },
        "interval": {
          "type": "between",
          "lowerBound": 180.0,
          "upperBound": 220.0
        }
      }
    ]
  }
}
\end{lstlisting}

\paragraph{Running the experiments.}
The implementation and experiment files are available in the accompanying
repository\footnote{\url{https://gitlab.com/vitnberg/wffa}}. 
The experiments are included in the repository under:
\begin{lstlisting}[language=bash,basicstyle=\ttfamily\scriptsize]
experiments/autocallable-structured-product
\end{lstlisting}
This directory contains the
JSON task files, shell scripts for executing the experiments, and generated result
files. The task files are grouped into payoff-analysis tests,
horizon-scalability tests, and constraint-scalability tests.

On Unix-like systems, all experiments can be executed from the experiment directory
as follows:
\begin{lstlisting}[language=bash,basicstyle=\ttfamily\scriptsize]
cd experiments/autocallable-structured-product
./scripts/run-all-tests.sh
\end{lstlisting}
If the scripts are not executable after checkout, run
\texttt{chmod +x scripts/*.sh} before executing them.

Individual experiment groups can be executed using the scripts in
\texttt{scripts/}. In particular,
\texttt{run-payoff-analysis-tests.sh} runs the payoff analysis reported in
Table~\ref{tab:extremal_analysis_results},
\texttt{run-horizon-scalability-tests.sh} runs the horizon-scalability experiment,
and \texttt{run-constraint-scalability-tests.sh} runs the constraint-scalability
experiment. The script \texttt{run-analysis-tasks.sh} executes the underlying JSON
analysis tasks directly.

The generated outputs are written to the \texttt{results/} directory. They include
the computed best-case and worst-case payoffs, witness event histories, product
automaton sizes, and runtime measurements used in the tables and figures of the
case study.

\paragraph{Output.}
Each experiment writes a textual report to the \texttt{results/} directory. The
report records the executed task file, the WFFA registry used for resolving model
references, the number of repeated runs, the computed extremal payoffs, witness
event histories, product-automaton metrics, and runtime statistics. A shortened
example is shown below.

\begin{lstlisting}[basicstyle=\ttfamily\scriptsize]
Task: experiments/autocallable-structured-product/tasks/payoff-analysis-tests/
      pattern-scenario-analysis-task-s1-t08.json
Registry: experiments/autocallable-structured-product/tasks/wffa/wffa-registry.json
Runs: 10

Pattern-based Scenario Analysis
===============================

Best case:
 - payoff: 126.0000
 - witness: FLAT FLAT UP DOWN FLAT UP UP AUTOCALL

Worst case:
 - payoff: 56.0000
 - witness: FLAT FLAT UP DOWN FLAT FLAT DOWN DOWN EXPIRATION

Metrics:
 - EHA states: 21
 - WFFA states: 77
 - product states: 1617
 - product transitions: 2226

Timing:
 - runs: 10
 - average: 87.546 ms
 - min: 53.219 ms
 - max: 264.188 ms
\end{lstlisting}

\end{document}